\documentclass[12pt]{article}%
\usepackage{amsfonts}
\usepackage{amsmath}
\usepackage{revsymb}
\usepackage{graphicx}
\usepackage{geometry}
\usepackage{wrapfig}
\usepackage{amssymb}
\usepackage{hyperref}
\usepackage[square,sort&compress,comma]{natbib}
\usepackage{hypernat}%
\providecommand{\U}[1]{\protect\rule{.1in}{.1in}}
\newtheorem{theorem}{Theorem}

\newtheorem{corollary}[theorem]{Corollary}

\newtheorem{definition}[theorem]{Definition}

\newtheorem{lemma}[theorem]{Lemma}

\newenvironment{proof}[1][Proof]{\noindent\textbf{#1.} }{\ \rule{0.5em}{0.5em}}

\begin{document}

\title{Two-sided estimates of minimum-error distinguishability of mixed quantum
states via generalized Holevo-Curlander bounds}
\author{Jon Tyson\thanks{jonetyson@X.Y.Z, where X=post, Y=Harvard, Z=edu}\\Jefferson Lab, Harvard University}
\date{November 27, 2008\\
Journal ref: \href{http://link.aip.org/link/?JMAPAQ/50/032106/1}{J. Math.
Phys. \textbf{50}, 032106 (2009)} }
\maketitle

\begin{abstract}
\noindent We prove a concise factor-of-two estimate for the failure-rate of
optimally distinguishing an arbitrary ensemble of mixed quantum states,
generalizing work of Holevo [Theor. Probab. Appl. \textbf{23}, 411 (1978)] and
Curlander [Ph.D. Thesis, MIT, 1979]. A modification of the minimal principle
of Concha and Poor [\textit{Proceedings of the 6th International Conference on
Quantum Communication, Measurement, and Computing} (Rinton, Princeton, NJ,
2003)] is used to derive a sub-optimal measurement which has an error rate
within a factor of two of the optimal \textit{by construction}. This
measurement is quadratically weighted, and has appeared as the first iterate
of a sequence of measurements proposed by Je\v{z}ek, \v{R}eh\'{a}\v{c}ek, and
Fiur\'{a}\v{s}ek [Phys. Rev. A \textbf{65}, 060301]. Unlike the so-called
\textquotedblleft pretty good\textquotedblright\ measurement, it coincides
with Holevo's asymptotically-optimal measurement in the case of
non-equiprobable pure states. A quadratically-weighted version of the
measurement bound by Barnum and Knill [J. Math. Phys. \textbf{43}, 2097
(2002)] is proven. Bounds on the distinguishability of syndromes in the sense
of Schumacher and Westmoreland [Phys. Rev. A \textbf{56}, 131 (1997)] appear
as a corollary. An appendix relates our bounds to the trace-Jensen inequality.

\end{abstract}

\pagebreak

\section{Introduction}

The \textit{minimum-error quantum distinguishability problem}\ is of obvious
practical importance in the design of optical detectors \cite{Helstrom Quantum
Detection and Estimation Theory} and of fundamental importance in subject of
quantum information \cite{pure state HSW theorem, mixed state HSW theorem,
Holevo mixed state HSW theorem} and quantum computation \cite{Ip Shor's
algorithm is optimal, Bacon, Childs from optimal to efficient algo, optimal
alg for hidden shift, moore and russels distinguishing, Bacon new hidden
subgroup}:

\begin{quote}
If an unknown state $\rho_{k}$ is randomly chosen from a known ensemble of
quantum states, what is the chance that the value of $k$ will be discovered by
an optimal measurement?
\end{quote}

\noindent Although various necessary and sufficient conditions for optimal
measurements have been derived
\cite{Yuen Ken Lax Optimum testing of
multiple,Holevo optimal measurement conditions 1,Holevo remarks on optimal
measurements, Belavkin optimal distinction of non-orthogonal quantum signals,
Belavkin Optimal multiple quantum statistical hypothesis testing, Belavkin and
Vancjan, Barnett and Croke On the conditions for discrimination between
quantum states with minimum error}
(see also \cite{eldar short and sweet optimal measurements}), and a number of
numerical algorithms for computing optimal measurements have been implemented
\cite{eldar short and sweet optimal measurements, Helstrom Bayes cost
reduction, Jezek Rehacek and Fiurasek Finding optimal strategies for minimum
error quantum state discrimination, Hradil et al Maximum Likelihood methods in
quantum mechanics, Tyson Simplified and robust conditions for optimum testing
of multiple hypotheses}, it is unlikely that an explicit general solution is
forthcoming. A number of works give interesting general upper and/or lower
bounds on quantum distinguishability
\cite{pure state HSW theorem,Hayden Leung Multiparty
Hiding,Barnum Knill UhOh,Montanaro on the distinguishability of random quantum
states,Daowen Qui Minimum-error discrimination between mixed quantum
states,Montanaro a lower bound on the probability of error in quantum state
discrimination, Hayashi Kawachi Kobayashi quantum measurements for hidden
subgroup problems with optimal sample complexity, Qiu and Li bounds on the minimum error discrimination between
mixed quantum states}%
.

The theory of optimal measurements has been generalized in several directions,
including to Belavkin and Maslov's theory of wave discrimination
\cite{Belavkin Book} and to the theory of optimal quantum channel reversals,
in the sense of average entanglement fidelity \cite{Barnum Knill UhOh,
Fletcher Thesis Channel Adapted quantum error correction,Fletcher Shor Win
Optimum quantum error recovery using semidefinite programming,Fletcher Shor
Win Channel-Adapted quantum error correction for the amplitude dampin
channel,Fletcher Shor Win Structured Near-Optimal Channel-Adapted Quantum
Error Correction, Taghavi Channel-Optimized quantum error correction}. Success
rates of optimal measurements have recently been expressed in terms of the
conditional min-entropy of bipartite classical-quantum states in Theorem 1 of
\cite{Konig Renner Schaffner Operational meaning of min and max entropy}. In
particular, the problem of finding the conditional min-entropy of an arbitrary
bipartite quantum state generalizes the optimal distinguishability problem.

\subsection{Results}

Theorem $\ref{Theorem mixed state Holevo interval}$ of section
$\ref{Section mixed state distinguishibility results}$ combines ideas of
Holevo \cite{Holovo Assym Opt Hyp Test}, Curlander \cite{Curlander thesis
MIT}, and Concha \& Poor \cite{Concha Thesis, Concha and Poor mixed optimal in
lease squares sense, Concha and Poor book chapter on Advances in quantum
detection} to give mathematically concise upper and lower distinguishability
bounds for arbitrary ensembles of mixed quantum states. Employing an
approximate minimal principle, a suboptimal measurement is derived which has a
failure rate within a factor of two of the optimal \textit{by construction}.
This measurement is observed to be the first iteration in the sequence of
measurements of Je\v{z}ek, \v{R}eh\'{a}\v{c}ek, and
Fiur\'{a}\v{s}ek.\cite{Jezek Rehacek and Fiurasek Finding optimal strategies
for minimum error quantum state discrimination} In the case of pure states
this measurement reduces to Holevo's asymptotically-optimal measurement
\cite{Holovo Assym Opt Hyp Test}, which is the quadratically-weighted Belavkin
square root measurement.\cite{Belavkin optimal distinction of non-orthogonal
quantum signals, Belavkin Optimal multiple quantum statistical hypothesis
testing, prequel}

Theorem $\ref{Theorem result Barnum Knill Proof Montanaro}$ of section
$\ref{Section Curlander's upper bound}$ combines ideas of Curlander and Holevo
to give somewhat tighter distinguishability bounds. Furthermore, a
quadratically-weighted version of Barnum and Knill's measurement bounds
\cite{Barnum Knill UhOh} are obtained, as are bounds on distinguishability of
syndromes in the sense of Schumacher and Westmoreland.

To motivate our considerations in the case of mixed states, section
$\ref{section Holevo's pure state bound}$ revisits Holevo and Curlander's
pure-state bounds.

Future directions appear in the final section. The appendix relates our bounds
to the trace-Jensen inequality.

\section{Holevo-Curlander pure state distinguishability bounds
\label{section Holevo's pure state bound}}

Before attempting to distinguish elements of mixed-state ensembles, it is
instructive to revisit a pure-state bound used by Holevo in his proof of the
asymptotic optimality Theorem (Theorem
\ref{Theorem Holevo's Assymtotic optimality theorem}, below).

\begin{theorem}
[Holevo \cite{Holovo Assym Opt Hyp Test}]The ensemble $\mathcal{E}%
_{m}=\left\{  \left(  \psi_{k},p_{k}\right)  \right\}  _{k=1,\ldots,m}$ of
linearly-independent pure states has the following minimum-error
distinguishability bound:%
\begin{equation}
P_{\text{fail}}^{\text{optimal}}\leq2\gamma_{\text{Holevo}}\text{,}
\label{holevo upper bound}%
\end{equation}
where%
\begin{equation}
\gamma_{\text{Holevo}}=1-\operatorname*{Tr}\sqrt{%
{\displaystyle\sum_{k}}
p_{k}^{2}\left\vert \psi_{k}\right\rangle \left\langle \psi_{k}\right\vert
}\text{,} \label{qholevo formula}%
\end{equation}
and $P_{\text{fail}}^{\text{optimal}}$ is the failure rate of the optimal measurement.
\end{theorem}

A two-sided version of Holevo's bound $\left(  \ref{holevo upper bound}%
\right)  $%
\begin{equation}
P_{\text{fail}}^{\text{optimal}}\in\left[  \gamma_{\text{Holevo}}%
,2\gamma_{\text{Holevo}}\right]
\text{,\label{Holevo pure state two-sided bound}}%
\end{equation}
was proved using different techniques by Curlander \cite{Curlander thesis MIT}
under the additional assumption of equiprobability ($p_{k}=1/m$). Equation
$\left(  \ref{Holevo pure state two-sided bound}\right)  $ follows without
Curlander's restriction by the following trivial modification of Holevo's argument.

Holevo restricted attention to orthonormal von Neumann measurement bases
$\left\{  e_{k}\right\}  $,\footnote{Rank-1 projective measurements are
optimal for distinguishing linearly-independent pure states.\cite{kennedy
linearly independent implies von Neumann,Belavkin optimal distinction of
non-orthogonal quantum signals,Helstrom Quantum Detection and Estimation
Theory,Mochon PGM}} with phases chosen so that
\begin{equation}
\left\langle e_{k},\psi_{k}\right\rangle \geq0. \label{Holevo positive ip
condition}%
\end{equation}
Instead of minimizing the probability of failure
\begin{equation}
P_{\text{fail}}\left(  \left\{  e_{k}\right\}  \right)  =\sum_{k=1}^{m}%
p_{k}\left(  1-\left\vert \left\langle e_{k},\psi_{k}\right\rangle \right\vert
^{2}\right)  ,
\end{equation}
Holevo considered the tractable approximation%
\begin{equation}
C_{\text{Holevo}}\left(  \left\{  e_{k}\right\}  \right)  =\sum_{k=1}^{m}%
p_{k}\left\Vert \psi_{k}-e_{k}\right\Vert ^{2}\text{,}
\label{pure state holevo cost}%
\end{equation}
which is equation 8 of \cite{Holovo Assym Opt Hyp Test}$.$ Since the phase
condition $\left(  \ref{Holevo positive ip condition}\right)  $ implies that%
\begin{equation}
1-\left\vert \left\langle \psi_{k},e_{k}\right\rangle \right\vert ^{2}%
=\frac{\left\Vert e_{k}-\psi_{k}\right\Vert ^{2}}{2}\left(  1+\left\langle
\psi_{k},e_{k}\right\rangle \right)  \in\left[  1/2,1\right]  \times\left\Vert
e_{k}-\psi_{k}\right\Vert ^{2}, \label{Holevo low level factor of two bound}%
\end{equation}
one has%
\begin{equation}
P_{\text{fail}}\left(  \left\{  e_{k}\right\}  \right)  \in\left[
1/2,1\right]  \times C_{\text{Holevo}}\left(  \left\{  e_{k}\right\}  \right)
, \label{holevo
minimized both sides of this}%
\end{equation}
where we use the notation $\left[  a,b\right]  \times c=\left[  ac,bc\right]
$. The bound $\left(  \ref{Holevo pure state two-sided bound}\right)  $
follows from minimization of $C_{\text{Holevo}}$. The minimizer is the
(usually sub-optimal) measurement basis\footnote{The final section of Holevo's
paper contains minor algebra errors. A corrected version of the minimizer
$\left(  \ref{Holevo measurement vector}\right)  $ appears in
\cite{EldarSquareRootMeasurement}, which also removes Holevo's assumption of
linear-independence. (This generalization may also be accomplished simply by
using Naimark's theorem, as advocated by Kebo \cite{Kebo thesis}.) Holevo's
measurement $\left(  \ref{Holevo measurement vector}\right)  $ belonged to the
previously-considered class of Belavkin weighted square root measurements.
\cite{Belavkin optimal distinction of non-orthogonal quantum signals, Belavkin
Optimal multiple quantum statistical hypothesis testing}, also called
\textquotedblleft weighted least-squares measurements\textquotedblright%
\ \cite{EldarSquareRootMeasurement} and \textquotedblleft generalized `pretty
good' measurements\textquotedblright\ \cite{Mochon PGM}.}%
\begin{equation}
e_{k}^{\text{Holevo}}=\left(
{\displaystyle\sum}
p_{\ell}^{2}\left\vert \psi_{\ell}\right\rangle \left\langle \psi_{\ell
}\right\vert \right)  ^{-1/2}p_{k}\left\vert \psi_{k}\right\rangle \text{.}
\label{Holevo measurement vector}%
\end{equation}

\section{Definitions, background, and notation:\label{definitions section}}

In this section we collect the technical definitions and mathematical
background needed for the rest of the paper. Throughout we shall consider an
ensemble
\begin{equation}
\mathcal{E}_{m}=\left\{  \left(  \rho_{k},p_{k}\right)  \right\}
_{k=1,\ldots,m}%
\end{equation}
of quantum states $\rho_{k}$ on a Hilbert space $\mathcal{H}$ with
\textit{a-priori} probabilities $p_{k}$, with $%
{\displaystyle\sum}
p_{k}=1$ and $\operatorname*{Tr}\rho_{k}=1.$ One may take $m=\infty$ without
changing our results. For the special case of pure states, $\rho_{k}$ will be
denoted by $\rho_{k}=\left\vert \psi_{k}\right\rangle \left\langle \psi
_{k}\right\vert $.

\begin{definition}
The ensemble $\mathcal{E}_{m}$ is \textbf{equiprobable} if $p_{k}=1/m$ for all
$k$. The subspace $\operatorname*{Span}\left(  \mathcal{E}_{m}\right)  $ is
the \textbf{span} of the ensemble $\mathcal{E}_{m}$, i.e. the span of the
ranges of the $p_{k}\rho_{k}$. A \textbf{positive-operator valued measure
(POVM) }(see, for example, p. 74 of \cite{Helstrom Quantum Detection and
Estimation Theory}) for distinguishing $\mathcal{E}_{m}$ is a collection of
positive semidefinite operators $\left\{  M_{k}\right\}  _{k=1,..,m}$ such
that $\sum M_{k}=%
\openone
_{\mathcal{H}}$ or $%
\openone
_{\operatorname*{Span}\left(  \mathcal{E}_{m}\right)  }$. The probability that
the value $i$ is detected when the POVM is applied to the state $\rho_{j}$ is
given by $p_{i|j}=\operatorname*{Tr}M_{i}\rho_{j}$. In particular, the
\textbf{success rate} for the POVM to distinguish the ensemble $\mathcal{E}%
_{m}$ is given by%
\begin{equation}
P_{\text{succ}}\left(  M_{k}\right)  =\sum_{k=1}^{m}p_{k}\operatorname*{Tr}%
\left(  \rho_{k}M_{k}\right)  =1-P_{\text{fail}}\text{.}%
\end{equation}
The \textbf{optimal success rate} is%
\begin{equation}
P_{\text{succ}}^{\text{optimal}}=\sup_{\text{POVMs\ }\left\{  M_{k}\right\}
}P_{\text{succ}}\left(  M_{k}\right)  =1-P_{\text{fail}}^{\text{optimal}%
}\text{.}%
\end{equation}

\end{definition}

A common POVM is

\begin{definition}
The \textbf{Belavkin-Hausladen-Wootters \textquotedblleft pretty good
measurement\textquotedblright\ (PGM)}\footnote{The PGM for non-equiprobable
pure states appeared in 1975 as an optimal measurement under conditions of
equality along the diagonal of the Graham matrix \cite{Belavkin optimal
distinction of non-orthogonal quantum signals, Belavkin Optimal multiple
quantum statistical hypothesis testing}$,$ and reappeared in 1993 as an
approximately-optimal measurement \cite{HausladenThesis, HausWootPGM}.} is
given by%
\begin{equation}
M_{k}=\left(
{\displaystyle\sum_{\ell=1}^{m}}
p_{\ell}\rho_{\ell}\right)  ^{-1/2^{+}}p_{k}\rho_{k}\left(
{\displaystyle\sum_{\ell=1}^{m}}
p_{\ell}\rho_{\ell}\right)  ^{-1/2^{+}}\text{,} \label{Def of mixed state PGM}%
\end{equation}
where one defines%
\begin{equation}
A^{-1/2^{+}}=%
{\displaystyle\sum_{\lambda_{j}>0}}
\lambda_{j}^{-1/2}\left\vert \psi_{j}\right\rangle \left\langle \psi
_{j}\right\vert ,
\end{equation}
for a spectral decomposition $A=%
{\displaystyle\sum}
\lambda_{j}\left\vert \psi_{j}\right\rangle \left\langle \psi_{j}\right\vert $.
\end{definition}

Numerical evidence \cite{Jezek Rehacek and Fiurasek Finding optimal strategies
for minimum error quantum state discrimination, Hradil et al Maximum
Likelihood methods in quantum mechanics} suggests that the following sequence
of measurements converges to the optimal measurement:

\begin{definition}
The \textbf{Je\v{z}ek-\v{R}eh\'{a}\v{c}ek-Fiur\'{a}\v{s}ek iterative
measurements} $\left\{  M_{k}^{\left(  n\right)  }\right\}  _{k=1,\ldots,m},$
$n\in\mathbb{Z}^{+}$, are recursively defined by \cite{Jezek Rehacek and
Fiurasek Finding optimal strategies for minimum error quantum state
discrimination, Hradil et al Maximum Likelihood methods in quantum mechanics}%
\begin{align}
M_{k}^{\left(  0\right)  }  &  =%
\openone
/m\text{ for }m<\infty\text{, }M_{k}^{\left(  0\right)  }=%
\openone
\text{ for }m=\infty\\
M_{k}^{\left(  n\right)  }  &  =\left(
{\displaystyle\sum_{\ell=1}^{m}}
p_{\ell}^{2}\rho_{\ell}M_{\ell}^{\left(  n-1\right)  }\rho_{\ell}\right)
^{-1/2^{+}}p_{k}^{2}\rho_{k}M_{k}^{\left(  n-1\right)  }\rho_{k}\left(
{\displaystyle\sum_{\ell=1}^{m}}
p_{\ell}^{2}\rho_{\ell}M_{\ell}^{\left(  n-1\right)  }\rho_{\ell}\right)
^{-1/2^{+}}\text{.} \label{Jezek iterate recursive def}%
\end{align}
If $\mathcal{E}_{m}$ is a pure-state ensemble then \textbf{Holevo's
measurement} \cite{Holovo Assym Opt Hyp Test} is given by
\begin{equation}
M_{k}=\left\vert e_{k}^{\text{Holevo}}\right\rangle \left\langle
e_{k}^{\text{Holevo}}\right\vert ,
\label{Holevo's measurement with an M sub k}%
\end{equation}
where $e_{k}^{\text{Holevo}}$ is given by $\left(
\ref{Holevo measurement vector}\right)  $. In particular, Holevo's measurement
is the pure-state version of the first
Je\v{z}ek-\v{R}eh\'{a}\v{c}ek-Fiur\'{a}\v{s}ek iterate.
\end{definition}

\noindent\textbf{Remark: }$M_{k}^{\left(  0\right)  }$ is also a POVM for
$m<\infty$. Note that since the recursion formula $\left(
\ref{Jezek iterate recursive def}\right)  $ is invariant under rescalings
$M_{k}^{\left(  n-1\right)  }\rightarrow\lambda\times M_{k}^{\left(
n-1\right)  }$, the cases $m<\infty$ and $m=\infty$ are essentially the same
for $n>0$.

Holevo studied measurements which were asymptotically optimal in the following
precise sense:

\begin{definition}
A measurement procedure $G$ is a mapping from ensembles to corresponding
POVMs. It is \textbf{asymptotically optimal \cite{Holovo Assym Opt Hyp Test}
}for distinguishing pure-state ensembles if for fixed $p_{1},...,p_{m}$ one
has
\begin{equation}
\frac{P_{\text{fail}}^{\text{G}\left(  \mathcal{E}_{m}\right)  }\left(
\mathcal{E}_{m}\right)  }{P_{\text{fail}}^{\text{opt}}\left(  \mathcal{E}%
_{m}\right)  }\rightarrow1
\end{equation}
as the states $\psi_{k}$ of $\mathcal{E}_{m}$ approach an orthonormal
set.\footnote{It is presumably intractable to produce a closed-form
measurement process $G$ for which $P_{\text{fail}}^{G\left(  \mathcal{E}%
_{m}\right)  }\left(  \mathcal{E}_{m}\right)  /P_{\text{fail}}^{\text{optimal}%
}\left(  \mathcal{E}_{m}\right)  \rightarrow1$ as the $\psi_{k}$ and $p_{k}$
are arbitrarily varied in such a way that $P_{\text{fail}}^{\text{optimal}%
}\left(  \mathcal{E}_{m}\right)  \rightarrow0$. Otherwise, one could recover
the optimal measurement for a fixed ensemble $\mathcal{E}_{m}$ on
$\mathcal{H}$ by taking the $\lambda\rightarrow1^{-}$ limit of the ensemble
$\mathcal{E}_{m+1}^{\prime}\equiv\left\{  \left(  \psi_{k},\left(
1-\lambda\right)  p_{k}\right)  \right\}  \cup\left\{  \left(  \phi
,\lambda\right)  \right\}  $ on a dilation $\mathcal{H}^{\prime}%
\supset\mathcal{H}$, with $\phi\bot\mathcal{H}$.}
\end{definition}

\noindent Holevo showed that

\begin{theorem}
[Holevo's asymptotic-optimality Theorem (1977) \cite{Holovo Assym Opt Hyp
Test}]\label{Theorem Holevo's Assymtotic optimality theorem}Holevo's
measurement $\left(  \ref{Holevo measurement vector}\right)  $ is
asymptotically optimal. Furthermore, for fixed $\left\{  p_{k}\right\}  $ one
has
\begin{equation}
\frac{2\gamma_{\text{Holevo}}\left(  \mathcal{E}_{m}\right)  }{P_{\text{fail}%
}^{\text{opt}}\left(  \mathcal{E}_{m}\right)  }\rightarrow1
\end{equation}
as $\left\langle \psi_{i},\psi_{j}\right\rangle \rightarrow\delta_{ij}$, where
$\gamma_{\text{Holevo}}$ is given by $\left(  \ref{qholevo formula}\right)  $.
\end{theorem}

\noindent A converse was proven in \cite{prequel}.

The following norms will be used:

\begin{definition}
Let $\mathcal{H}$ and $\mathcal{K}$ be Hilbert spaces, and let $A:\mathcal{H}%
\rightarrow\mathcal{K}$ be a bounded linear operator. The \textbf{absolute
value }is $\left\vert A\right\vert =\sqrt{A^{\dag}A}$. The \textbf{trace norm
}is $\left\Vert A\right\Vert _{1}=\operatorname*{Tr}\left\vert A\right\vert $.
The \textbf{Frobenius norm} is $\left\Vert A\right\Vert _{2}=\sqrt
{\operatorname*{Tr}A^{\dag}A}$. The \textbf{operator norm} is given by%
\begin{equation}
\left\Vert A\right\Vert =\sup_{0\neq\psi\in\mathcal{H}}\frac{\left\Vert
A\psi\right\Vert }{\left\Vert \psi\right\Vert }\text{.}%
\end{equation}
$A$ is an \textbf{isometry }if $A^{\dag}A=%
\openone
$.
\end{definition}

It will be assumed that the reader is familiar with the following properties
of the trace-norm, which may be found in \cite{Reed and Simon I}:

\begin{enumerate}
\item $\left\vert \operatorname*{Tr}A\right\vert \leq\left\Vert A\right\Vert
_{1}=\left\Vert A^{\dag}\right\Vert _{1}$

\item $\left\Vert WA\right\Vert _{1}\leq\left\Vert W\right\Vert \times
\left\Vert A\right\Vert _{1}$

\item If $\dim\mathcal{K}\geq\dim\mathcal{H}$ then
\begin{equation}
\sup_{\text{isometries }U:\mathcal{H}\rightarrow\mathcal{K}}\operatorname{Re}%
\left(  \operatorname*{Tr}A^{\dag}U\right)  =\left\Vert A\right\Vert _{1},
\label{equation sup trace A dag U}%
\end{equation}
where $U$ is a maximizer iff%
\begin{equation}
\left.  U\right\vert _{\operatorname{Ran}\left(  A^{\dag}A\right)  }=A\left(
A^{\dag}A\right)  ^{-1/2^{+}}\text{.\label{unitary maximizer equation}}%
\end{equation}

\end{enumerate}

\noindent\textbf{Note:} Property 3 is a simple consequence of the
singular-value decomposition.

\newpage

\section{Mixed-state distinguishability bounds using Holevo's
method\label{Section mixed state distinguishibility results}}

The first step in constructing a mixed-state version of the argument of
section $\ref{section Holevo's pure state bound}$ is to construct a
mixed-state version of the underlying estimate $\left(
\ref{Holevo low level factor of two bound}\right)  $:

\begin{lemma}
\label{lemma containing mixed state holevo proto estimate}Let $\rho$ be a
density matrix on $\mathcal{H}$ and let $E:\mathcal{H}\rightarrow\mathcal{H}$
be an operator with $\left\Vert E\right\Vert \leq1$. Then
\begin{equation}
1-\operatorname*{Tr}\left(  E^{^{\dagger}}E\rho\right)  \in\left[  1,2\right]
\times\left(  1-\left\Vert E\rho\right\Vert _{1}\right)
\label{eq mixed failure probability estimate}%
\end{equation}

\end{lemma}

\begin{proof}
The lower bound follows from the properties of the trace-norm:%
\[
1-\operatorname*{Tr}\left(  E^{\dag}E\rho\right)  \geq1-\left\Vert E^{\dag
}E\rho\right\Vert _{1}\geq1-\left\Vert E\right\Vert \times\left\Vert
E\rho\right\Vert _{1}\geq1-\left\Vert E\rho\right\Vert _{1}%
\]
To prove the upper bound, define the pre-inner product on the bounded
operators on $\mathcal{H}$ by%
\[
\left\langle E,F\right\rangle _{\rho}=\operatorname*{Tr}_{\mathcal{H}}\left(
E^{\dag}F\rho\right)  \text{.}%
\]
By Bessel's inequality%
\[
\operatorname*{Tr}\left(  E^{\dag}E\rho\right)  =\left\Vert E\right\Vert
_{\rho}^{2}\geq\sup_{U\text{ unitary}}\frac{\left\vert \left\langle
U,E\right\rangle _{\rho}\right\vert ^{2}}{\left\Vert U\right\Vert _{\rho}^{2}%
}=\sup\frac{\left\vert \operatorname*{Tr}U^{\dag}E\rho\right\vert ^{2}%
}{\left(  \operatorname*{Tr}\rho\right)  ^{2}}=\left\Vert E\rho\right\Vert
_{1}^{2}\text{.}%
\]
Subtracting both sides from $1$,%
\[
1-\operatorname*{Tr}\left(  E^{^{\dagger}}E\rho\right)  \leq\left(
1+\left\Vert E\rho\right\Vert _{1}\right)  \left(  1-\left\Vert E\rho
\right\Vert _{1}\right)  \leq2\left(  1-\left\Vert E\rho\right\Vert
_{1}\right)  \text{.}%
\]

\end{proof}

To find the measurement properly analogous to $\left(
\ref{Holevo measurement vector}\right)  ,$ one simply needs to minimize the
cost function arising from $\left(
\ref{eq mixed failure probability estimate}\right)  $:

\begin{theorem}
\label{Theorem mixed state Holevo interval}Let $M_{k}=E_{k}^{\dag}E_{k}$ be a
POVM on $\operatorname*{Span}\left(  \mathcal{E}_{m}\right)  $ minimizing the
approximate cost function%
\begin{equation}
C\left(  \left\{  E_{k}\right\}  \right)  =%
{\displaystyle\sum_{k}}
p_{k}\left(  1-\left\Vert E_{k}\rho_{k}\right\Vert _{1}\right)  \text{.}
\label{Approximate cost function}%
\end{equation}
Then $M_{k}$ is the first Je\v{z}ek-\v{R}eh\'{a}\v{c}ek-Fiur\'{a}\v{s}ek
iterate $\left(  \ref{Jezek iterate recursive def}\right)  $
\begin{equation}
M_{k}=\left(
{\displaystyle\sum}
p_{\ell}^{2}\rho_{\ell}^{2}\right)  ^{-1/2^{+}}p_{k}^{2}\rho_{k}^{2}\left(
{\displaystyle\sum}
p_{\ell}^{2}\rho_{\ell}^{2}\right)  ^{-1/2^{+}}\text{,}
\label{Equation for generalized holevo measurement in Theorem statement}%
\end{equation}
and%
\begin{equation}
\Gamma\leq P_{\text{fail}}^{\text{optimal}}\leq P_{\text{fail}}\left(
\left\{  M_{k}\right\}  \right)  \leq2\Gamma
\text{,\label{equation nice mixed state holevo curlander estimate in theorem statement}%
}%
\end{equation}
where%
\begin{equation}
\Gamma=\Gamma\left(  \mathcal{E}_{m}\right)  =\min_{\left\{  E_{k}\right\}
}C\left(  \left\{  E_{k}\right\}  \right)  =1-\operatorname*{Tr}\sqrt
{\sum_{k=1}^{m}p_{k}^{2}\rho_{k}^{2}}\in\left[  0,1\right)  \text{.}
\label{Formula for capital gamma}%
\end{equation}

\end{theorem}

\noindent\textbf{Remark: }The approximate cost function $\left(
\ref{Approximate cost function}\right)  $ is a somewhat-disguised modification
of the minimal principle of Concha and Poor \cite{Concha Thesis, Concha and
Poor mixed optimal in lease squares sense, Concha and Poor book chapter on
Advances in quantum detection}, which was reverse-engineered to reproduce the
\textit{ad-hoc} mixed-state PGM. The measurement $\left(
\ref{Equation for generalized holevo measurement in Theorem statement}\right)
$ is an example of a mixed-state Belavkin weighted measurement. (See section
2.2 of \cite{Belavkin Book}.) A discussion of the relative merits of various
weightings, including the cubic weighting of \cite{Wehner thesis, Wehner State
discrimination with post-measurement information}, may be found in
\cite{prequel}.\smallskip

\begin{proof}
By lemma $\ref{lemma containing mixed state holevo proto estimate}$,
$P_{\text{fail}}\left(  \left\{  E_{k}^{\dag}E_{k}\right\}  \right)
\in\left[  1,2\right]  \times C\left(  \left\{  E_{k}\right\}  \right)  $ for
all POVMs $M_{k}=E_{k}^{\dag}E_{k}$. Hence all that is required to get a
factor-of-two estimate of $P_{\text{fail}}^{\text{opt}}\left(  \mathcal{E}%
_{m}\right)  $ is to minimize $C$ subject to the constraint $%
{\displaystyle\sum}
E_{k}^{\dag}E_{k}=%
\openone
_{\operatorname*{Span}\left(  \mathcal{E}_{m}\right)  }$. Note that the
replacement $E_{k}\rightarrow W_{k}E_{k}$ for unitary $W_{k}$ does not alter
$C\left(  \left\{  E_{k}\right\}  \right)  $ or the quantities $E_{k}^{\dag
}E_{k}$. Hence the polar decomposition allows imposition of the additional
constraint $E_{k}\rho_{k}\geq0$, giving the expression%
\[
C\left(  \left\{  E_{k}\right\}  \right)  =\tilde{C}\left(  U\right)
=1-\operatorname*{Tr}V^{\dag}U,
\]
where $U,V:\operatorname*{Span}\left(  \mathcal{E}_{m}\right)  \rightarrow
\operatorname*{Span}\left(  \mathcal{E}_{m}\right)  \otimes\mathbb{C}^{m}$ are
defined by%
\begin{align*}
U\psi &  =\sum_{k=1}^{m}\left\vert E_{k}\psi\right\rangle \otimes\left\vert
k\right\rangle _{\mathbb{C}^{m}}\\
V\psi &  =\sum_{k=1}^{m}\left\vert p_{k}\rho_{k}\psi\right\rangle
\otimes\left\vert k\right\rangle _{\mathbb{C}^{n}}\text{.}%
\end{align*}
Here $\left\vert k\right\rangle _{\mathbb{C}^{m}}$ is the standard basis of
$\mathbb{C}^{m}$. Note that $U$ is an isometry iff $M_{k}=E_{k}^{\dag}E_{k}$
is a POVM on $\operatorname*{Span}\left(  \mathcal{E}_{m}\right)  $. By
equation $\ref{equation sup trace A dag U}$,%
\[
\min_{\text{isometries }U}\tilde{C}\left(  U\right)  =1-\left\Vert
V\right\Vert _{1}=1-\operatorname*{Tr}\sqrt{\sum_{k=1}^{m}p_{k}^{2}\rho
_{k}^{2}},
\]
with minimizer%
\[
U=V\left(  V^{\dag}V\right)  ^{-1/2^{+}}.
\]
This gives%
\[
E_{k}^{\text{min}}=\left\langle k\right\vert _{\mathbb{C}^{m}}U=p_{k}\rho
_{k}\left(  \sum_{\ell=1}^{m}p_{\ell}^{2}\rho_{\ell}^{2}\right)  ^{-1/2^{+}%
}\text{.}%
\]
Since $E_{k}^{\text{min}}\rho_{k}\geq0$, the theorem follows.
\end{proof}

\newpage

\section{Generalization of Curlander's upper
bound\label{Section Curlander's upper bound}}

The upper bound of $\left(
\ref{equation nice mixed state holevo curlander estimate in theorem statement}%
\right)  $ may be sharpened by combining Holevo's measurement $\left(
\ref{Holevo measurement vector}\right)  $ with Curlander's argument of Ref.
\cite{Curlander thesis MIT}\textit{:}

\begin{theorem}
\label{Theorem result Barnum Knill Proof Montanaro}The optimal failure rate
for distinguishing the arbitrary mixed-state ensemble $\mathcal{E}%
_{m}=\left\{  \left(  \rho_{k},p_{k}\right)  \right\}  _{k=1,\ldots,m}$
satisfies%
\begin{equation}
\Gamma\leq P_{\text{fail}}^{\text{opt}}\leq P_{\text{fail}}^{\text{HJRF}}%
\leq\Gamma\left(  2-\Gamma\right)  \leq2\Gamma
\text{,\label{bound curlander-holevo interval}}%
\end{equation}
where $P_{\text{fail}}^{\text{HJRF}}$ is the failure rate of the measurement
$\left(
\ref{Equation for generalized holevo measurement in Theorem statement}\right)
$ and $\Gamma=\Gamma\left(  \mathcal{E}_{m}\right)  $ is given by $\left(
\ref{Formula for capital gamma}\right)  $. Furthermore,
\begin{equation}
P_{\text{fail}}^{\text{opt}}\leq P_{\text{fail}}^{\text{HJRF}}\leq\left(
1+P_{\text{succ}}^{\text{opt}}\right)  P_{\text{fail}}^{\text{opt}}\text{.}
\label{Barnum-knill bound for holevo weighting}%
\end{equation}

\end{theorem}

\noindent\textbf{Note:} Curlander proved $\left(
\ref{bound curlander-holevo interval}\right)  $ in the special case of
equiprobable pure states.\cite{Curlander thesis MIT} Barnum and Knill have
already shown that the bound $\left(
\ref{Barnum-knill bound for holevo weighting}\right)  $ holds for the
mixed-state \textquotedblleft pretty good\textquotedblright\ measurement
\cite{Barnum Knill UhOh}. Note that the RHS of $\left(
\ref{Barnum-knill bound for holevo weighting}\right)  $ never exceeds $1$, so
the bound is always meaningful.

\bigskip

\noindent\textbf{NOTE\ ADDED\ TO\ ARXIV\ VERSION: }It was not realized at the
time of publication that the \textit{lower} bound of $\left(
\ref{bound curlander-holevo interval}\right)  $ admits a generalization using
the theory of matrix monotonicity \cite{TysonMon}. Furthermore, this
generalization is a minor variation of a similar bound of \cite{Ogawa and
Nagoaka Strong converse to the quantum channel coding theorem}.\footnote{An
erratum or comment will be sent to JMP to this effect.}

\bigskip

\begin{proof}
First restrict consideration to pure-state ensembles $\mathcal{E}_{m}=\left\{
\left(  \psi_{k},p_{k}\right)  \right\}  _{k=1,\ldots,m}$. By the convexity of
$x\mapsto x^{2}$ and Jensen's inequality,\footnote{The author's argument,
which is similar to Curlander's, was originally movtivated by that used to
prove Lemma 2 of \cite{Montanaro on the distinguishability of random quantum
states}.}%
\begin{align}
P_{\text{succ}}^{\text{HJRF}}  &  =%
{\displaystyle\sum}
p_{k}\left\langle e_{k}^{\text{Holevo}},\psi_{k}\right\rangle ^{2}%
\label{insert PGM here for a mess}\\
&  \geq\left(
{\displaystyle\sum}
p_{k}\left\langle e_{k}^{\text{Holevo}},\psi_{k}\right\rangle \right)
^{2}\nonumber\\
&  =\left(
{\displaystyle\sum_{k}}
p_{k}\left\langle \psi_{k}\right\vert \left(
{\displaystyle\sum_{\ell}}
p_{\ell}^{2}\left\vert \psi_{\ell}\right\rangle \left\langle \psi_{\ell
}\right\vert \right)  ^{-1/2}p_{k}\left\vert \psi_{k}\right\rangle \right)
^{2}\nonumber\\
&  =\left(  1-\Gamma\right)  ^{2}\nonumber
\end{align}
so that%
\[
P_{\text{fail}}^{\text{HJRF}}\leq1-\left(  1-\Gamma\right)  ^{2}=\Gamma\left(
2-\Gamma\right)  \text{.}%
\]
The left-most inequality of $\left(  \ref{bound curlander-holevo interval}%
\right)  $ was already proved in Theorem
$\ref{Theorem mixed state Holevo interval}$.

In the more general case of mixed states, take spectral decompositions
$\rho_{k}=%
{\textstyle\sum\nolimits_{\ell}}
\mu_{k\ell}\left\vert \psi_{k\ell}\right\rangle \left\langle \psi_{k\ell
}\right\vert $ and consider the pure-state ensemble
\begin{equation}
\mathcal{E}_{m}^{\ast}=\left\{  \left(  \left\vert \psi_{k\ell}\right\rangle
,p_{k}\mu_{k\ell}\right)  \right\}  . \label{syndrome ensemble}%
\end{equation}
Note that any measurement $\left\{  M_{k\ell}\right\}  $ for $\mathcal{E}%
_{m}^{\ast}$ may be converted into a measurement $M_{k}=%
{\textstyle\sum\nolimits_{\ell}}
M_{k\ell}$ for $\mathcal{E}_{m}$, which trivially satisfies%
\[
P_{\text{fail}}\left(  \left\{  M_{k\ell}\right\}  \right)  \geq
P_{\text{fail}}\left(  \left\{  M_{k}\right\}  \right)  \text{.}%
\]
In particular, $\mathcal{E}_{m}^{\ast}$ is less distinguishable than
$\mathcal{E}_{m}$, and the measurement $\left(
\ref{Equation for generalized holevo measurement in Theorem statement}\right)
$ is less successful at distinguishing it. Then using $\left(
\ref{equation nice mixed state holevo curlander estimate in theorem statement}%
\right)  $ and the pure-state case,%
\begin{equation}
\Gamma\left(  \mathcal{E}_{m}\right)  \leq P_{\text{fail}}^{\text{opt}}\left(
\mathcal{E}_{m}\right)  \leq P_{\text{fail}}^{\text{HJRF}}\left(
\mathcal{E}_{m}\right)  \leq P_{\text{fail}}^{\text{HJRF}}\left(
\mathcal{E}_{m}^{\ast}\right)  \leq\Gamma\left(  \mathcal{E}_{m}\right)
\left(  2-\Gamma\left(  \mathcal{E}_{m}\right)  \right)
\text{.\label{chain of inequals in proof of mixed state curlander}}%
\end{equation}
Note that last inequality used the identity $\Gamma\left(  \mathcal{E}%
_{m}^{\ast}\right)  =\Gamma\left(  \mathcal{E}_{m}\right)  $.

Note that $\Gamma\in\left[  0,1\right)  $ by $\left(
\ref{Formula for capital gamma}\right)  $. Because $\gamma\mapsto\gamma\left(
2-\gamma\right)  $ is monotonic increasing on $\gamma\in\left[  0,1\right)  $,
the chain of inequalities $\left(
\ref{Barnum-knill bound for holevo weighting}\right)  $ follows by plugging in
the left-hand-side of the first inequality of $\left(
\ref{bound curlander-holevo interval}\right)  $ into the right-hand-side of
the third.
\end{proof}

\noindent\textbf{Remark:} In Schumacher and Westmoreland's classic paper
\cite{mixed state HSW theorem}, the elements of $\mathcal{E}_{m}$ appear as
\textquotedblleft codewords,\textquotedblright\ with \textquotedblleft
syndromes\textquotedblright\ given by elements of $\mathcal{E}_{m}^{\ast}$.
Schumacher and Westmoreland assert that measurements of $\mathcal{E}_{m}%
^{\ast}$ are \textquotedblleft not really more difficult\textquotedblright%
\ than measurements of $\mathcal{E}_{m}$. It is now easy to quantify this assertion:

\begin{corollary}
Let $\mathcal{E}_{m}^{\ast}$ be the ensemble $\left(  \ref{syndrome ensemble}%
\right)  $ of eigenvectors of the elements of $\mathcal{E}_{m}$. Then%
\[
P_{\text{fail}}^{\text{opt}}\left(  \mathcal{E}_{m}\right)  \leq
P_{\text{fail}}^{\text{opt}}\left(  \mathcal{E}_{m}^{\ast}\right)  \leq\left(
1+P_{\text{succ}}^{\text{opt}}\left(  \mathcal{E}_{m}\right)  \right)
P_{\text{fail}}^{\text{opt}}\left(  \mathcal{E}_{m}\right)  \text{.}%
\]

\end{corollary}

\begin{proof}
Simply replace the quantity $P_{\text{fail}}^{\text{HJRF}}\left(
\mathcal{E}_{m}\right)  $ by $P_{\text{fail}}^{\text{opt}}\left(
\mathcal{E}_{m}^{\ast}\right)  $ in the chain of inequalities $\left(
\ref{chain of inequals in proof of mixed state curlander}\right)  $, and
continue as in the proof of $\left(  \ref{bound curlander-holevo interval}%
\right)  $.
\end{proof}

\section{Reflections on the quadratic weighting}

As we have seen, the quadratic weighting gives rise to some particularly
simple bounds for distinguishability of quantum states. For comparison,
substituting the linearly-weighted \textquotedblleft pretty
good\textquotedblright\ measurement $\left(  \ref{Def of mixed state PGM}%
\right)  $ into equation $\left(  \ref{insert PGM here for a mess}\right)  $
gives the upper bound%
\[
P_{\text{fail}}^{\text{PGM}}\leq1-\operatorname*{Tr}\left(  \left(
{\displaystyle\sum_{\ell}}
p_{\ell}\left\vert \psi_{\ell}\right\rangle \left\langle \psi_{\ell
}\right\vert \right)  ^{-1/2}\left(
{\displaystyle\sum_{\ell}}
p_{\ell}^{3/2}\left\vert \psi_{\ell}\right\rangle \left\langle \psi_{\ell
}\right\vert \right)  \right)  .
\]

The relative simplicity of the quadratic bound $\left(
\ref{bound curlander-holevo interval}\right)  $ is not surprising. As shown by
the author in \cite{prequel}, Holevo's pure-state measurement $\left(
\ref{Holevo measurement vector}\right)  $ has the following conceptual and
practical advantages over the \textit{ad-hoc} \textquotedblleft pretty
good\textquotedblright\ measurement:\footnote{It is of course assumed that the
\textit{a-priori} probabilities $p_{k}$ are not all the same, or Holevo's
pure-state measurement and the PGM\ would be identical.}

\begin{enumerate}
\item Holevo's asymptotic-optimality property uniquely specifies Holevo's
measurement among the class of Belavkin weighted measurements.

\item Holevo's measurement categorically outperforms the PGM for ensembles of
two pure states.

\item The optimality conditions for Holevo's measurement are particularly simple.
\end{enumerate}

\noindent The previous sections provide more examples of this theme:

\begin{enumerate}
\item[4.] The quadratically-weighted mixed-state measurement gives
particularly simple pure- and mixed-state distinguishability bounds.

\item[5] The approximate cost function $\left(
\ref{Approximate cost function}\right)  $ for the quadratic measurement is
within a factor of two of the function $P_{\text{fail}}\left(  M_{k}\right)
$. (The corresponding cost functions for the pure and mixed-states PGMs
\cite{EldarSquareRootMeasurement, Concha Thesis, Concha and Poor mixed optimal
in lease squares sense, Concha and Poor book chapter on Advances in quantum
detection} admit no such comparison.)
\end{enumerate}

\section{Conclusion and Future Directions}

As we have seen, mathematically concise (and reasonably tight) bounds on the
distinguishability of mixed quantum states may be obtained by combining the
ideas of Holevo, Curlander, and Concha \& Poor. In the above we have not
explained the connection between these ideas and the iterative algorithm of
Je\v{z}ek, \v{R}eh\'{a}\v{c}ek, and Fiur\'{a}\v{s}ek, other than to recognize
that a natural generalization of Holevo's argument gives the first iterate of
Je\v{z}ek \textit{et al}'s measurements.

A proper setting to explore such questions is in the theory of approximate
quantum channel reversals, which Barnum and Knill \cite{Barnum Knill UhOh}
have already investigated using a generalization of the \textquotedblleft
pretty good\textquotedblright\ measurement. We will consider an abstract form
of JRF iteration, study its convergence properties, and construct bounds on
channel reversibility and relative min-entropy in future work \cite{More
coming, TQC2009talk}. We will also attempt to reconsider Holevo's notion of
asymptotic optimality in this setting.

\bigskip

\noindent\textbf{Acknowledgements: }I would like to thank Aram Harrow, Julio
Concha, V. P. Belavkin, and Vincent Poor for pointing out useful references,
Julio Concha and Andrew Kebo for providing copies of their theses, and William
Wootters for providing a copy of Hausladen's thesis, and Stephanie Wehner for
a valuable discussion. I would also like to thank Arthur Jaffe, Chris King,
and Peter Shor for their encouragement and the editors and anonymous referees
for their useful comments and suggestions.

\section*{Appendix A: An application of the trace-Jensen inequality}

The following theorem makes it transparent that $1-\operatorname*{Tr}%
\sqrt{\sum p_{k}^{2}\rho_{k}^{2}}\geq0$, giving some insight into the bounds
$\left(
\ref{equation nice mixed state holevo curlander estimate in theorem statement}%
\right)  $ and $\left(  \ref{bound curlander-holevo interval}\right)  $:

\begin{theorem}
\label{theorem trace jensen shows mixed holevo curlander bound is positive}Let
$f:\left[  0,\infty\right)  \rightarrow\mathbb{R}$ be concave with $f\left(
0\right)  =0$, and consider positive semidefinite operators $A_{k}$ on a
Hilbert space $\mathcal{H}$. Then
\[
\operatorname*{Tr}f\left(  \sum_{k=1}^{N}A_{k}\right)  \leq\operatorname*{Tr}%
\sum_{k=1}^{N}f\left(  A_{k}\right)  \text{,}%
\]
where $f\left(  A\right)  $ is defined using the functional calculus
\cite{Reed and Simon I}. (In particular, $f\left(  A\right)  =%
{\displaystyle\sum}
f\left(  \lambda_{i}\right)  \left\vert \psi_{i}\right\rangle \left\langle
\psi_{i}\right\vert $ for any spectral decomposition $A=%
{\displaystyle\sum}
\lambda_{i}\left\vert \psi_{i}\right\rangle \left\langle \psi_{i}\right\vert $.)
\end{theorem}

\begin{proof}
The case $N=2$ is sufficient. By the trace-Jensen inequality
\cite{HansenPedersenTraceJensen}
\begin{align}
\operatorname*{Tr}f\left(  A\right)   &  =\operatorname*{Tr}f\left(
E_{1}^{\dagger}\left(  A+B\right)  E_{1}+E_{2}^{\dagger}0E_{2}\right)
\nonumber\\
&  \geq\operatorname*{Tr}\left(  E_{1}^{\dagger}f\left(  A+B\right)
E_{1}+E_{2}^{\dagger}f\left(  0\right)  E_{2}\right) \nonumber\\
&  =\operatorname*{Tr}\left(  A^{1/2}\left(  A+B\right)  ^{-1/2^{+}}f\left(
A+B\right)  \left(  A+B\right)  ^{-1/2^{+}}A^{1/2}\right)
\tag{A1}
\label{add cylic 1}%
\end{align}
where
\begin{align*}
E_{1}^{\dagger}  &  =A^{1/2}\left(  A+B\right)  ^{-1/2^{+}}\\
E_{2}  &  =\sqrt{1-E_{1}^{\dagger}E_{1}}.
\end{align*}
Similarly,
\begin{equation}
\operatorname*{Tr}f\left(  B\right)  \geq\operatorname*{Tr}\left(
B^{1/2}\left(  A+B\right)  ^{-1/2^{+}}f\left(  A+B\right)  \left(  A+B\right)
^{-1/2^{+}}B^{1/2}\right)  . \tag{A2}\label{add cyclic 2}%
\end{equation}
The conclusion follows by adding $\left(  \ref{add cylic 1}\right)  $ and
$\left(  \ref{add cyclic 2}\right)  $ and applying the cyclicity of the trace.
\end{proof}

\bigskip

\noindent\textbf{Note added in proof:} The \textit{lower} bound $\Gamma\left(
\xi\right)  \leq P_{\text{fail}}\left(  M_{k}^{\text{opt}}\right)  $ of
Theorem 10 admits a simple generalization proved using matrix monotonicity:
\[
1-\operatorname*{Tr}\left[  \left(
{\textstyle\sum\nolimits_{k=1}^{m}}
p_{k}^{s}\rho_{k}^{s}\right)  ^{1/s}\right]  \leq P_{\text{fail}}\left(
M_{k}^{\text{opt}}\right)  ,
\]
for any $s\in\left[  1,\infty\right)  $. This is addressed in a short note
which has been submitted to this journal.\cite{TysonMon}

\end{document}